\documentclass[journal,twoside,web]{ieeecolor}
\usepackage{lcsys}
\usepackage{amsmath,amssymb,amsfonts}
\usepackage{algpseudocode}
\usepackage{algorithm}
\usepackage{graphicx}
\usepackage{graphics}
\usepackage{textcomp}
\usepackage{subcaption}
\usepackage{hyperref}
\usepackage{cleveref}
\usepackage[font=small,skip=0pt]{caption}

\usepackage[backend=bibtex]{biblatex}
\AtNextBibliography{\small}

\def\BibTeX{{\rm B\kern-.05em{\sc i\kern-.025em b}\kern-.08em
    T\kern-.1667em\lower.7ex\hbox{E}\kern-.125emX}}
\newtheorem{theorem}{Theorem}

\newtheorem{remark}{Remark}

\newtheorem{assumption}{Assumption}

\newcommand{\revise}[1]{{\color{black} \noindent #1}}

\begin{document}
\title{Generalized Model Predictive Path Integral Control as Expectation--Maximization}

\author{Jiarui Wang, Sina Sharifi and Mahyar Fazlyab, \IEEEmembership{Member, IEEE}
\thanks{J. Wang, S. Sharifi, and M. Fazlyab are with the Department of Electrical and Computer Engineering at Johns Hopkins University, Baltimore, MD 21218, USA.
{\tt\small \{jwang486, sshari12, mahyarfazlyab\}@jhu.edu}}
\thanks{
This work was supported by the National Science Foundation (NSF) under Grant 2515978 and by the U.S. DOE, Office of Science, ASCR program under the Scientific Discovery through Advanced Computing (SciDAC) Institute “LEADS: LEarning-Accelerated Domain Science”.
}
}

\maketitle
\thispagestyle{empty}

\begin{abstract}
Model Predictive Path Integral (MPPI) control is a powerful sampling-based method for solving stochastic optimal control problems and has enabled real-time control in complex robotic systems. Despite its empirical success, its theoretical understanding remains limited. In this work, we show that MPPI can be interpreted as a special case of the Expectation–Maximization (EM) algorithm applied to a probabilistic inference formulation of optimal control. This perspective leads to a generalized EM-MPPI framework that extends MPPI beyond the commonly used Gaussian parameterization. We analyze the convergence behavior of this algorithm and characterize the local convergence rate in terms of the covariance of the posterior trajectory distribution and the exploration distribution. For exponential-family distributions, we establish a sufficient increase property of the log-likelihood when the log-partition function is strongly convex. Specializing the analysis to Gaussian MPPI yields explicit global and local convergence characterizations. 
\end{abstract}


\section{Introduction}\label{sec:introduction}
\IEEEPARstart{M}{odel} Predictive Path Integral (MPPI) \cite{williams2016aggressive, williams2017information} is a widely used sampling-based method for nonlinear optimal control. Its strong empirical performance, together with modern parallel simulation tools such as Mujoco \cite{todorov2012mujoco} and Isaac Gym \cite{makoviychuk2021isaac}, has enabled real-time control in a broad range of robotic settings, including whole-body control \cite{alvarez2025real}, quadrotor navigation \cite{zhai2025pa}, 
and online policy adaptation \cite{wang2024residual}. 
MPPI updates control sequences by sampling noisy trajectories, evaluating their costs, and forming a weighted average of sampled controls, resulting in a simple update rule that can be efficiently parallelized.

Most existing derivations of MPPI are rooted in stochastic optimal control and control-as-inference viewpoints \cite{williams2016aggressive, williams2017information, honda2025model}, in which trajectory optimality is recast as an inference problem over trajectories. These perspectives explain the exponential weighting structure underlying the MPPI update, but they do not by themselves provide a general optimization-theoretic characterization of the algorithm or its convergence behavior. Recent works have started to study MPPI convergence properties \cite{yi2024covo, homburger2025optimality}. 
However, a general optimization perspective that links the MPPI update to a classical iterative optimization algorithm remains lacking.

\textbf{Contributions.}
In this work, we show that MPPI can be interpreted as an instance of the Expectation–Maximization (EM) algorithm. This perspective yields a unified probabilistic and optimization-theoretic interpretation of MPPI, and naturally suggests algorithmic generalizations beyond the standard Gaussian setting. Building on this perspective, we analyze the convergence of a general EM-MPPI algorithm and establish a monotonic improvement condition for exponential-family distributions with strongly convex log-partition functions. We then specialize the analysis to standard MPPI, for which we derive explicit convergence guarantees and characterize the resulting rate in closed form.
{\color{blue}
The code is available \href{https://github.com/o4lc/MPPI-EM}{here}.
}

\section{Related Work}
\label{section:related-work}
\textbf{Control as Probabilistic Inference.}
Stochastic optimal control can be formulated as probabilistic inference by
introducing optimality variables whose likelihood depends exponentially on
trajectory cost or reward \cite{kappen2012optimal,honda2025model}. This yields
a posterior trajectory distribution proportional to a prior distribution
Boltzmann-reweighted by the negative cost, and optimal control can be interpreted
as computing or approximating expectations under this posterior. This viewpoint
has been developed in control as inference \cite{levine2018reinforcement} and
variational-inference MPC \cite{okada2020variational}, where  MPPI \cite{williams2016aggressive}, the Cross-Entropy Method (CEM) \cite{botev2013cross}, and
\revise{Covariance Matrix Adaptation Evolution Strategy (CMA-ES)} \cite{hansen2003reducing} arise under different choices of trajectory distributions and update
rules. 
\revise{
Closely related to our work, \cite{wang2021variational} derives a
weighted maximum-likelihood projection onto a parametric variational family;
For mixture-of-Gaussian families, this projection can be solved using an inner
EM algorithm. In contrast, we identify the exact expectation-based MPPI update
itself as an EM iteration.
}

\textbf{Expectation--Maximization Methods for Control.}
The EM algorithm \cite{dempster1977maximum,mclachlan2008algorithm}, originally
developed for latent-variable inference \cite{koller2009probabilistic}, has
also been applied to reinforcement learning and optimal control.
\cite{dayan1997using} interpreted certain reinforcement-learning updates as EM
by treating actions as latent variables and rewards as observations.
\cite{toussaint2011expectation} proposed an EM approach for solving
\revise{(Partially Observed) Markov Decision Processes} by introducing a random horizon,
so that trajectories and horizons are optimized as latent variables. EM has
also been used for stochastic optimal control under linear--Gaussian dynamics
and policies \cite{mallick2022stochastic}. These works demonstrate the value
of EM in control, but they do not identify standard MPPI itself as an EM
iteration or analyze MPPI convergence from this viewpoint. Our work fills this
gap by deriving MPPI as an EM-type algorithm over a parameterized control
distribution and using this interpretation to establish convergence guarantees.

\textbf{Theoretical Analysis of MPPI.}
Despite MPPI's empirical success, its theoretical analysis remains limited.
Existing works mainly study the standard sampling-based update. For example,
\cite{yi2024covo} establishes linear convergence rates for quadratic problems,
including time-varying LQR, and extends the analysis to nonlinear systems.
\cite{homburger2025optimality} studies optimality and suboptimality properties
of MPPI in stochastic and deterministic settings. These results provide
important guarantees for the standard MPPI update under assumptions on the
objective, dynamics, and sampling covariance. In contrast, we identify the EM
structure underlying MPPI and use it to derive convergence guarantees from a
latent-variable maximum-likelihood perspective.

\section{Preliminaries}

\subsection{Optimal Control and MPC}
Consider the finite-horizon optimal control problem
\begin{align}\label{eq:optimal-control}
    \min_{u} ~& J(u) = \sum_{h=1}^H c_h(x_h, u_h) + c_f(x_{H+1}) \notag \\
        \text{s.t.} ~& x_{h+1} = f_h(x_h, u_h), \quad h=1,\dots,H.
\end{align}
Here $u = [u_1, \dots, u_H]$ denotes the control sequence, $f_h: \mathbb{R}^{n_x} \times \mathbb{R}^{n_u} \rightarrow \mathbb{R}^{n_x}$ represents the system dynamics, $c_h :\mathbb{R}^{n_x} \times \mathbb{R}^{n_u} \rightarrow \mathbb{R}$ is the stage cost, and $c_f :\mathbb{R}^{n_x} \rightarrow \mathbb{R}$ is the terminal cost. We assume $J$ is bounded below with optimal value $J^* = \min_{u} J(u)$.
In MPC, problem \eqref{eq:optimal-control} is repeatedly solved in a receding-horizon manner. At each time step, the current state is used as the initial state $x_1$, an optimal control sequence is computed, and only the first control action $u_1$ is applied to the system. \revise{The horizon is then shifted forward, and the optimization problem is solved again (usually warm-started by shifting the control sequence obtained at the previous MPC step) after measuring the next state.}

\subsection{MPPI}
While nonlinear programming methods are widely used for solving \eqref{eq:optimal-control}, their applicability may be limited when the dynamics $f_h$ are non-differentiable or available only through a simulator. For example, in contact-rich systems, the dynamics may itself be defined implicitly through an optimization problem, making gradient-based methods difficult to apply.
In such scenarios, MPPI control provides a sampling-based alternative. Starting from a nominal control sequence \(u\), MPPI generates candidate control sequences by sampling perturbations around \(u\), evaluates their rollout costs, and updates the nominal sequence via a cost-weighted average:
\begin{align}\label{eq:MPPI}
    u_+ =
    \frac{
        \sum_{i=1}^N u[i] \exp \big(-\frac{J(u[i])}{\tau}\big)
    }{
        \sum_{i=1}^N \exp \big(-\frac{J(u[i])}{\tau}\big)
    },
    \quad 
    u[i] \overset{\text{i.i.d.}}{\sim} \mathcal{N}(u,\Sigma),
\end{align}
where $\tau>0$ is a temperature parameter and $\{u[i]\}_{i=1}^N$ are sampled control sequences. Each $u[i] = [u[i]_1, \dots, u[i]_H]$ represents a control trajectory drawn from the Gaussian distribution centered at the current nominal control sequence $u$.
In practice, the update \eqref{eq:MPPI} may be repeated multiple times within a single MPC step before applying the first control action, as implemented in practical systems (e.g., ROS2 Nav2 MPPI controller\footnote{\url{https://docs.ros.org/en/iron/p/nav2_mppi_controller/}}).

\section{Generalized MPPI from EM}
%
To formulate problem \eqref{eq:optimal-control} as an inference problem, instead of directly searching for a single optimal control sequence $u$, we consider a distribution $p(u;\theta)$ parameterized by $\theta$. 
We introduce an optimality variable $\mathcal{O}$ conditioned on $\mathcal{U}$, defined as a Bernoulli random variable with conditional probability
\begin{align}
    P(\mathcal{O}=1 \mid \mathcal{U}=u)
    =
    \exp \big(
        -\frac{J(u)-J^*}{\tau}
    \big),
\end{align}
\revise{where $\tau>0$ is a temperature parameter controlling the sharpness of the
optimality likelihood: smaller $\tau$ concentrates probability on trajectories
with costs close to $J^*$, while larger $\tau$ gives broader weight to
near-optimal trajectories.} This construction ensures $P(\mathcal{O}=1 \mid \mathcal{U}=u)\in (0,1]$ and assigns higher probability to sequences with lower cost.
The marginal log likelihood of the optimality event under $p(u;\theta)$ is
{\small
\begin{align} \label{eq:log-likelihood}
    \log P(\mathcal{O}=1;\theta)
    &= 
    \log \int
        p(u;\theta)
        \exp\big(
            -\frac{J(u)-J^*}{\tau}
        \big) du
        \notag \\
    &=
    \log \int
        p(u;\theta)
        \exp\big(
            -\frac{J(u)}{\tau}
        \big) du
        + \text{const}.
\end{align}
}
Provided the integral is finite, this defines a well-posed objective over \(\theta\), and the constant is independent of $\theta$. 
\revise{
This objective becomes large when $p(u;\theta)$ assigns more probability to control sequences that have lower trajectory costs.}
Therefore, instead of solving \eqref{eq:optimal-control} by directly optimizing a single control sequence, we optimize \(\theta\) so that the induced sampling distribution increasingly favors low-cost trajectories.
This maximum-likelihood viewpoint naturally leads to an EM procedure for updating \(\theta\).

%
In the probabilistic formulation above, the control sequence \(\mathcal U\) plays the role of a latent variable, while the optimality event \(\mathcal O=1\) is the observed variable. The goal is to maximize the marginal log likelihood \(\ell(\theta):=\log P(\mathcal O=1;\theta)\) with respect to \(\theta\). Since this is a latent-variable maximum-likelihood problem, it can be addressed using the EM algorithm.
To derive the EM updates, let \(q(u)\) be an auxiliary distribution over control sequences. Then from \eqref{eq:log-likelihood} we have
{\small
\begin{align}\label{eq:ELBO}
    \ell(\theta) \! 
    =&
    \log P(\mathcal{O}=1;\theta) \notag \\
    =& \log \int p(u;\theta) \exp(-\frac{J(u)}{\tau}) du + \text{const} \notag\\
    =& \log \int q(u) \frac{p(u;\theta) \exp(-\frac{J(u)}{\tau})}{q(u)} du + \text{const} \notag\\
    =& \log  \mathbb{E}_{u \sim q(u)} \bigg[ 
        \frac{p(u;\theta) \exp(-\frac{J(u)}{\tau})}{q(u)} 
        \bigg] \! + \! \text{const} \notag\\
    \geq& \mathbb{E}_{u \sim q(u)} \bigg[ 
        \log \frac{p(u;\theta) \exp(-\frac{J(u)}{\tau})}{q(u)} 
        \bigg] \! + \! \text{const} \notag\\
    =& \mathbb{E}_{u \sim q(u)} \big[ -\frac{J(u)}{\tau}\big] 
        + \mathbb{E}_{u \sim q(u)} \big[ \log \frac{p(u;\theta)}{q(u)} \big] + \text{const} \notag\\
    =& -\frac{1}{\tau} \bigg\{ 
        \mathbb{E}_{u \sim q(u)}[J(u)] + \tau KL(q \| p(\cdot;\theta))
    \bigg\} + \text{const},
\end{align}
}
where the inequality follows from Jensen's inequality applied to the concave logarithm. \eqref{eq:ELBO} is the evidence lower bound (ELBO), which is tight when \(q(u)\) is chosen as the posterior distribution of \(\mathcal U\) conditioned on the observation \(\mathcal O=1\).

\subsection{E-Step}
For fixed \(\theta\), the ELBO in \eqref{eq:ELBO} is maximized over \(q\) when Jensen's inequality is tight, i.e., when
$p(u;\theta) \exp(-\frac{J(u)}{\tau})/q(u)$ is constant in $u$. Using Bayes' rule, the optimal auxiliary distribution that maximizes the ELBO for the current $\theta$ is
\begin{align} \label{eq:E-step-q-gibbs}
    q(u;\theta) 
    &= p(u \mid \mathcal{O}=1;\theta) =\frac{
        p(u;\theta) \exp(-\frac{J(u)}{\tau})
    }{
        \int p(u;\theta) \exp(-\frac{J(u)}{\tau}) du
    }.
\end{align}
Therefore, for a fixed $\theta$, the ELBO \eqref{eq:ELBO} is maximized over $q$ by choosing $q(u)=q(u;\theta)$, which corresponds to the posterior distribution $p(u\mid \mathcal{O}=1;\theta)$.
\revise{
Thus, the E-step recovers the standard optimal-control posterior obtained
by Boltzmann reweighting of the current proposal distribution.}

\subsection{M-Step}
Given the ELBO constructed in the previous section, which is tight at the current $\theta$, 
the M-step updates the parameter by maximizing the ELBO with respect to $\theta$ \revise{ to find the next iterate $\theta_+$}, while fixing $q$ obtained from the E-step. Since the term \(\mathbb E_{u\sim q}[J(u)]\) in \eqref{eq:ELBO} does not depend on \(\theta_+\), this is equivalent to minimizing the Kullback--Leibler divergence between $q$ and the parametric distribution $p(\cdot;\theta_+)$:
\begin{align}\label{eq:weighted-mle-expectation}
    &\min_{\theta_+} KL\big(q(\cdot;\theta)\,\|\,p(\cdot;\theta_+)\big) \notag\\
    \Longleftrightarrow \;
    &\max_{\theta_+} \int q(u;\theta)\log p(u;\theta_+)\,du.
\end{align}
\revise{
Thus, the M-step recovers the standard variational projection step: it
approximates the Boltzmann-reweighted optimal-control distribution by the
next proposal \(p(\cdot;\theta_+)\) within the chosen parametric family.
}
Substituting the expression of $q(u;\theta)$ from the E-step yields
\begin{align}
    \max_{\theta_+}
    \frac{
        \int p(u;\theta)\exp(-\frac{J(u)}{\tau})\log p(u;\theta_+)\,du
    }{
        \int p(u;\theta)\exp(-\frac{J(u)}{\tau})\,du
    }.
\end{align}
which can be approximated using Monte Carlo estimation as 
\begin{align} \label{eq:wmle}
    &\max_{\theta_+} 
        \frac
        {\sum_{i=1}^N \exp(-\frac{J(u[i])}{\tau}) \log p(u[i]; \theta_+)}
        {\sum_{i=1}^N \exp(-\frac{J(u[i])}{\tau})} \notag\\
    \Longleftrightarrow 
    &\max_{\theta_+}  \sum_{i=1}^N w[i] \log p(u[i]; \theta_+).
\end{align}
\eqref{eq:wmle} is a weighted maximum likelihood estimate, where $u[i]$ are sampled from $p(u;\theta)$ and the weights $w[i]$ are computed as
\begin{align} \label{eq:MPPI-weights}
     w[i] = \frac{\exp(-\frac{J(u[i])}{\tau})}{\sum_{j=1}^N \exp(-\frac{J(u[j])}{\tau})}.
\end{align}
Combining the E-step and M-step yields the complete EM algorithm, which we refer to as \emph{Generalized MPPI}, summarized in \Cref{alg:EM-MPPI}. Note that although the E-step has the closed-form solution \eqref{eq:E-step-q-gibbs}, the posterior distribution is never explicitly computed.  Instead, it is implicitly used in the M-step through Monte Carlo sampling and weight computation.
Importantly, the algorithm naturally generalizes to any parametric distribution $p(\cdot;\theta)$ provided that samples can be efficiently drawn from $p(\cdot;\theta)$ and the weighted maximum likelihood problem \eqref{eq:wmle} can be efficiently solved.

\begin{algorithm}[t]
\caption{Generalized MPPI}
\label{alg:EM-MPPI}
\begin{algorithmic}[1]
    \Require: Initial parameter $\theta_0$, number of samples $N$
    \For{$k=0,...,$}
        \State Sample $u[1], ..., u[N] \overset{\text{i.i.d.}}{\sim} p(u;\theta_{k})$
        \State Compute weights using \eqref{eq:MPPI-weights}
        \State Obtain $\theta_{k+1}$ by solving \eqref{eq:wmle}
    \EndFor
\end{algorithmic}
\end{algorithm}

\subsection{Convergence}
The EM interpretation of MPPI induces the fixed-point iteration
\begin{align}\label{eq:EM-fixed-point}
    \theta_+ &= M(\theta) = \mathrm{argmax}_{\theta_+} Q(\theta_+; \theta) \\
    Q(\theta_+; \theta) &= \mathbb{E}_{u \sim q(u;\theta)}[\log p(u;\theta_+)]. \notag
\end{align}
Thus, MPPI can be analyzed through the convergence theory of EM-type algorithms.

\begin{assumption}\label{ass:EM}
    In this paper, we make the following assumptions
    \begin{enumerate}
        \item $\theta \in \Omega$ where $\Omega \subseteq \mathbb{R}^{n_p}$ is a closed set,
        \item $\forall u \text{ s.t. }  \|u\| < \infty, \ \lim_{\|\theta\| \rightarrow \infty} p(u;\theta) = 0$, 
        \revise{\item $\forall M > 0$, there exists a constant $C_M < \infty$ such that $p(u;\theta) \le C_M,
            \quad
            \forall \theta \in \Omega,\ \forall u \in \mathbb{R}^d
            \text{ with } \|u\| \le M,$}
        \item $p(u;\theta)$ is twice continuously differentiable in $\theta$,
        \item $\lim_{M \rightarrow \infty}  \inf \{J(u) : \|u\| \geq M\} = \infty$,
        \item \eqref{eq:EM-fixed-point} admits a maximizer for every $\theta$.
    \end{enumerate}
\end{assumption}

\revise{
Assumption~1 is mainly a regularity condition on the sampling family and the
trajectory cost. The coercivity condition on \(J\) (Assumption 1.5) is also natural, since large control sequences are either directly penalized
through control effort or lead to large terminal/running costs. The existence of the M-step maximizer (Assumption 1.6) is mild in the fixed-covariance Gaussian case, where the weighted maximum-likelihood problem has a closed-form solution (see \eqref{eq:gaussiam_mppi_closed_form_m_step}).  Importantly, Assumption~1 does not require differentiability of the dynamics
or of the rollout cost \(J\); it only requires smoothness of the chosen
sampling density \(p(u;\theta)\) with respect to its parameter. Thus, the
analysis remains compatible with nonsmooth or simulator-defined robotic tasks,
including navigation with obstacle penalties or contact-rich interactions.
}

\begin{theorem}[Convergence of EM-MPPI] \label{thm:global-convergence}
Let \Cref{ass:EM} hold and $\{\theta_k\}$ be the sequence generated by EM-MPPI. Then the following properties hold:
\begin{enumerate}
    \item $\ell(\theta_{k+1}) > \ell(\theta_k) \ \forall \theta_k$ that is not a stationary point,
    \item the sequence $\{\ell(\theta_k)\}$ converges,
    \item every limit point $\theta^*$ of $\{\theta_k\}$ is a stationary point of the log-likelihood, i.e., $\nabla_\theta \ell(\theta^*) = 0 $,
\end{enumerate}
\end{theorem}
\begin{proof}
    Under \Cref{ass:EM}, we can show that the super level set $S=\{\theta: \ell(\theta) \geq \ell(\theta_0) \}$ is compact, $\ell(\theta)$ is continuously differentiable, and that $Q(\theta_+,\theta)$ is continuous in both argument. Then the result follows directly from classical convergence results for the EM algorithm \cite{wu1983convergence, mclachlan2008algorithm}, since the proposed EM-MPPI update is an instance of an EM iteration. A more detailed proof of the compactness of $S$ will be available in the longer version.
\end{proof}
Furthermore, the EM-MPPI iteration enjoys a local linear convergence rate.
\begin{theorem}[Local linear convergence] \label{thm:local-linear-convergence}
Assume that \Cref{ass:EM} holds, and further suppose that
\begin{enumerate}
    \item $\nabla^2_{\theta}\ell(\theta^*) \prec 0$,
     \item the optimization problem \eqref{eq:EM-fixed-point} has a unique solution in a neighborhood of $\theta^*$,
     \item $\nabla^2_{11}Q(\theta^*, \theta^*)$ is nonsingular,
    \item $M(\theta)$ is continuously differentiable in a neighborhood of $\theta^*$,
\end{enumerate}
then there exists a neighborhood $\mathcal{N}$ of $\theta^*$ such that 
for any $\theta_0 \in \mathcal{N}$, the sequence 
$\{\theta_k\}$ generated by EM-MPPI converges to $\theta^*$. 
Moreover, the convergence is locally linear:
\begin{align*}
    \lim_{k \rightarrow \infty}
    \frac{\|\theta_{k+1} - \theta^*\|}{\|\theta_k - \theta^*\|}
    = \rho(\partial M(\theta^*)),
\end{align*}
where the Jacobian of the EM mapping is
\begin{align*}
\partial M(\theta^*)
&=
- \left(\nabla^2_{11} Q(\theta^*,\theta^*)\right)^{-1}
\nabla^2_{12} Q(\theta^*,\theta^*), \\
\nabla_{11}^2 Q(\theta^*,\theta^*)
&=
E_{q(u;\theta^*)}\!\left[\nabla_\theta^2 \log p(u;\theta^*)\right], \\
\nabla_{12}^2 Q(\theta^*,\theta^*)
&=
\mathrm{Cov}_{q(u;\theta^*)}
\left[\nabla_\theta \log p(u;\theta^*)\right],
\end{align*}
with $\rho(\partial M(\theta^*)) < 1$. Here $\rho(\partial M(\theta^*))$ denotes the spectral radius of the Jacobian of the EM mapping at $\theta^*$, \revise{and $\nabla_{ij}^2 Q$ denotes the second derivative of $Q$ with respect to its $i$-th and $j$-th arguments.}

\end{theorem}

The local convergence result follows directly from \cite[Sec~3.9]{mclachlan2008algorithm} and Ostrowski Theorem \cite[Thm 10.1.3]{ortega2000iterative}. A more detailed proof will be available in the longer version.

\revise{\begin{remark}[Interpretation of the local assumptions]
The additional assumptions in Theorem~2 are local nondegeneracy conditions for
the EM fixed-point map that describe the behavior of EM-MPPI once the iterates enter a neighborhood of a locally
attracting solution. In the fixed-covariance Gaussian case, the M-step is locally unique because the weighted maximum-likelihood problem is strictly concave in the mean parameter. The nonsingularity of
\(\nabla^2_{11}Q(\theta^*,\theta^*)\) rules out degenerate local curvature, while differentiability of \(M\) holds locally when the weighted moments vary smoothly with the current proposal parameter. Thus, Theorem~1 gives a global
monotonic-improvement interpretation under broad regularity conditions, whereas
Theorem~2 characterizes the local linear rate near a nondegenerate fixed point.
\end{remark}}

\section{EM-MPPI for Exponential Family} \label{sec:MPPI-exp-family}

A particularly important case is when $p(u;\theta)$ belongs to the exponential family, i.e.,
\begin{align*}
    p(u;\theta) = h(u)\exp \big( \eta(\theta)^\top T(u) - A(\eta(\theta))\big),
\end{align*}
where $h(u)$ is the base density, $\eta(\theta)$ is the natural parameter, $T(u)$ denotes the sufficient statistics, and $A(\eta)$ is the log-partition function, which is convex in $\eta$.

\revise{Assuming the mapping between $\theta$ and natural parameter $\eta$ is invertible, instead of studying convergence in $\theta$, we can analyze convergence in the natural parameter $\eta$.} Two useful identities for exponential families are
\begin{align}\label{eq:exp-family-log-partition}
    \nabla_\eta A(\eta) \!=\! \mathbb{E}_{p(u;\eta)}[T(u)],  \ 
    \nabla^2_\eta  A(\eta) \!=\! \mathrm{Cov}_{p(u;\eta)}[T(u)].
\end{align}
Note that when $p(u;\theta)$ belongs to the exponential family, the joint distribution $p(u,\mathcal{O}=1;\theta)$ also has an exponential-family form with a modified base density. This observation significantly simplifies the analysis.
The EM iteration \eqref{eq:EM-fixed-point} can be written as
\begin{align}\label{eq:exp-family-m-step}
    \eta_+ &= \mathrm{argmax}_{\eta_+} \mathbb{E}_{q(u;\eta)}[\eta_+^\top T(u) - A(\eta_+)] \\
    q(u;\eta) \!
    &= \!p(u|\mathcal{O}\!=\!1;\eta) 
    \!\propto \!
    h(u)\exp\big(
        \eta^\top T(u) - A(\eta) - \frac{J(u)}{\tau} 
    \big)\notag
\end{align}
where $q(u;\eta)$ is the distribution obtained from the E-step.
Since $A(\eta)$ is convex in $\eta$, the optimization problem \eqref{eq:exp-family-m-step} is convex. Combining the optimality condition with \eqref{eq:exp-family-log-partition} yields
\begin{align}\label{eq:exp-family-m-step-moment-matching}
    \mathbb{E}_{p(u;\eta_+)}[T(u)] = \nabla_{\eta} A(\eta_+) = \mathbb{E}_{q(u;\eta)}[T(u)],
\end{align}
which shows that the M-step corresponds to moment matching in the sufficient statistics $T(u)$.

\begin{theorem}[Sufficient Increase in Exponential Family]\label{thm:exp-family-sufficient-increase}
Assume $A(\eta)$ is $\alpha$-strongly convex in $\eta$. Then the EM iteration satisfies
\begin{align*}
    \ell(\eta_+) - \ell(\eta) \geq \frac{\alpha}{2}\|\eta_+ - \eta\|_2^2.
\end{align*}
\end{theorem}

\begin{proof}
See appendix for the proof.
\end{proof}

Note that this sufficient increase property also holds when \Cref{alg:EM-MPPI} performs only a single update per iteration, which is commonly done in practice.

\revise{
\begin{remark}[Monte Carlo approximation]
Theorem~\ref{thm:exp-family-sufficient-increase} is stated for the population EM update. In Algorithm~1, the exact
posterior moment \(\mathbb E_{q(u;\eta)}[T(u)]\) is replaced by its self-normalized estimate.
Repeating the proof of Theorem~\ref{thm:exp-family-sufficient-increase} shows that the sufficient-increase bound
holds up to an additional term proportional to the estimation error in this
moment. Under finite-second-moment assumptions, this error is
\(O(N^{-1/2})\). Thus, the deterministic guarantee is recovered in the
large-sample limit, while finite-sample monotonicity is approximate.
\end{remark}
}

\section{Special Cases}

\subsection{Gaussian MPPI}

In this section, we study the convergence of the commonly used Gaussian MPPI and derive an explicit convergence result using the results from previous sections. In Gaussian MPPI, the control distribution is
\begin{align*}
    p(u;\theta) = \frac{1}{(2\pi)^{n/2} |\Sigma|^{1/2}} 
    \exp \big(
        -\frac{1}{2}
        (u-\mu)^\top \Sigma^{-1} (u-\mu)
    \big),
\end{align*}
where $\theta=\{\mu\}$ is the parameter and $\Sigma$ is a fixed covariance matrix. 
\revise{ In this case, $p(u;\theta)$ is smooth in $\mu$, locally bounded on compact sets of controls, and satisfies \(p(u;\theta)\to0\) for each fixed \(u\) as \(\|\theta\|\to\infty\).}
The weighted log-likelihood estimation \eqref{eq:wmle} admits a closed-form solution
\begin{align}\label{eq:gaussiam_mppi_closed_form_m_step}
    \mu_+ = \sum_{i=1}^N w_i\, u[i],
\end{align}
which reduces exactly to the standard MPPI update.
Using the result from \Cref{thm:local-linear-convergence}, the Jacobian of the EM mapping at a stationary point is
\(
    \partial M(\theta^*) = \mathrm{Cov}_{q(u;\theta^*)}[u] \Sigma^{-1},
\)
which characterizes the local convergence rate. This shows that the local convergence is governed by the product of the covariance of the importance-weighted distribution and the inverse exploration covariance.

Next we study the global sufficient increase property of Gaussian MPPI. A Gaussian distribution with fixed covariance can be written in exponential-family form as
\begin{align*}
    h(u) &= \exp\!\left(-\tfrac12 u^\top \Sigma^{-1}u\right),
    \quad \eta(\theta) = \Sigma^{-1}\mu, \\
    T(u) &= u, \quad A(\eta) = \tfrac12 \eta^\top \Sigma \eta 
              + \tfrac12 \log\det(2\pi\Sigma).
\end{align*}
\revise{The mapping between $\mu$ and $\eta$ is therefore invertible.} Moreover, the log-partition function $A(\eta)$ is $\alpha$-strongly convex with $\alpha=\lambda_{\min}(\Sigma)$ where $\lambda_{\min}(\Sigma)$ is the minimum eigenvalue of $\Sigma$.
Applying \Cref{thm:exp-family-sufficient-increase}, Gaussian MPPI satisfies the global sufficient increase
\begin{align}\label{eq:gaussian-mppi-sufficient-increase}
    \ell(\eta_+) - \ell(\eta) 
    \geq 
    \frac{\lambda_{\min}(\Sigma)}{2}
    \|\eta_+ - \eta\|_2^2.
\end{align}

\subsection{Mixture of Gaussian (MoG) MPPI}

A well-known limitation of Gaussian MPPI is its inability to represent multi-modal distributions. For example, in a car navigation scenario with obstacle avoidance, going left and going right around an obstacle may both be valid choices, while their average trajectory may collide with the obstacle. 
Motivated by this issue, we consider a mixture of Gaussian distributions
\begin{align*}
    p(u;\theta) = 
    \sum_{l=1}^L \pi_l \mathcal{N}(u; \mu_l, \Sigma),
    \qquad 
    \sum_{l=1}^L \pi_l = 1,
\end{align*}
where $L$ is the number of Gaussian components and $\theta = \{\pi_1,\mu_1,\dots,\pi_L,\mu_L\}$ denotes the parameters. For simplicity, we assume that the covariance matrices are fixed and identical across components.
In this case, the weighted maximum likelihood estimation \eqref{eq:wmle} can be solved using an inner EM algorithm with the following update rule:
\begin{align*}
    r_{il}
    &=
    \frac{
        \pi_l \mathcal{N}(u[i];\mu_l,\Sigma)
    }{
        \sum_{k=1}^L \pi_k \mathcal{N}(u[i];\mu_k,\Sigma)
    } \\
    \pi_l^+
    &= 
    \frac{\sum_{i=1}^N w_i r_{il}}{\sum_{i=1}^N w_i}, \quad 
    \mu_l^+
    =
    \frac{\sum_{i=1}^N w_i r_{il} u[i]}
    {\sum_{i=1}^N w_i r_{il}}.
\end{align*}

Mixture-of-Gaussian control distributions for MPPI have also been explored in prior work on variational-inference control \cite{wang2021variational}. This formulation naturally fits into the EM-MPPI framework, where the E-step computes trajectory weights and the M-step updates the mixture parameters.


\section{Experiments}
In this section, we consider a Dubins car navigation problem with obstacle avoidance
\begin{align*}
    \min_{U} \;& \min \{\|x_h - d\|_Q^2 + \|u_h\|_R^2 : h=1,...,H\}\\
    \text{s.t.}\;& x_{h+1} = f(x_h, u_h), \quad \forall h=1,\dots,H,
\end{align*}
where $Q = \mathrm{diag}(1,1,0.01)$, $R=0.001$, and $x_h = [\mathtt{x}_h, \mathtt{y}_h, \mathtt{\theta}_h]$ contains the planar position and orientation, and $d$ is that goal state. The control input is $u_h = [\omega_h]$, representing the angular velocity, with actuation bound $\omega_h \in [-\frac{3}{2}\pi, \frac{3}{2}\pi]$. The system dynamics are
\[
    x_{h+1}
    = x_h + \delta t [ v \cos(\mathtt{\theta}_h), v \sin(\mathtt{\theta}_h), \omega_h]^\top,
\]
where $v$ is a constant forward velocity and $\delta t$ is the Euler integration step. We sample at least 1024 feasible trajectories for each update.

\subsection{Sufficient Increase of Gaussian MPPI}
We first verify the sufficient increase property of Gaussian MPPI in \eqref{eq:gaussian-mppi-sufficient-increase} by considering a simple navigation problem without obstacles. The car starts from the state $x_1 = [0,0,\pi/2]$ and the target is $d=[2,1,0]$. We run \Cref{alg:EM-MPPI} for 20 iterations. In \Cref{fig:car-simple}, we show the sampled trajectories at different iterations. We also plot the function
\(
    F(\eta_k) = \log ( \tfrac{1}{N} \sum_{i=1}^N \exp (-\tfrac{J(u[i])}{\tau})),
\)
which provides a Monte Carlo estimate of the likelihood (up to a constant), and
$
    F(\eta_k) + \frac{\lambda_{\min}(\Sigma)}{2} 
    \|\eta_{k+1} - \eta_k\|_2^2,
$
which is the lower bound on $F(\eta_{k+1})$ implied by \eqref{eq:gaussian-mppi-sufficient-increase}.

\begin{figure}[t]
\centering
\begin{subfigure}[t]{0.48\linewidth}
    \centering
    \includegraphics[width=\linewidth,height=3.4cm,keepaspectratio]{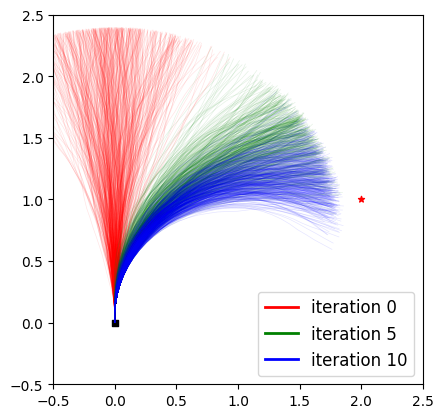}
    \caption{Sampled trajectories at different iterations}
    \label{fig:car_simple_trajs}
\end{subfigure}
\hfill
\begin{subfigure}[t]{0.49\linewidth}
    \centering
    \includegraphics[width=\linewidth,height=3.4cm,keepaspectratio]{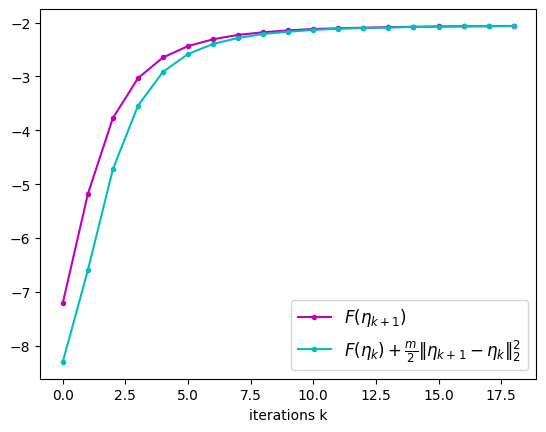}
    \caption{Estimate of log-likelihood}
    \label{fig:car-simple-convergence}
\end{subfigure}
\vspace{1mm}
\caption{
  Convergence of Gaussian MPPI on a simple Dubins-car navigation problem.
  }

\label{fig:car-simple}
\end{figure}

\subsection{Cluttered Environment Navigation}

Next, we consider a more cluttered environment where the car must reach the goal while avoiding multiple obstacles. We compare Gaussian MPPI and MoG MPPI with two Gaussian components. For both methods, we perform either 1 or 5 iterations of \Cref{alg:EM-MPPI} before executing each action. 
It can be seen in \Cref{fig:car-cluttered} that, as the car approaches an obstacle, the feasible trajectories split into two modes corresponding to passing the obstacle from the left or the right. MoG MPPI naturally captures these modes through its mixture distribution. In contrast, Gaussian MPPI approximates the bimodal distribution with a single Gaussian and therefore averages the two directions. As a result, the mean control tends to keep the car moving straight toward the obstacle, producing many infeasible samples. \Cref{table:car-cluttered-stats} shows the total number of sampled trajectories and the fraction of rejected trajectories.

\begin{figure}[t]
  \begin{subfigure}[t]{0.23\textwidth}
    \includegraphics[scale=0.38]{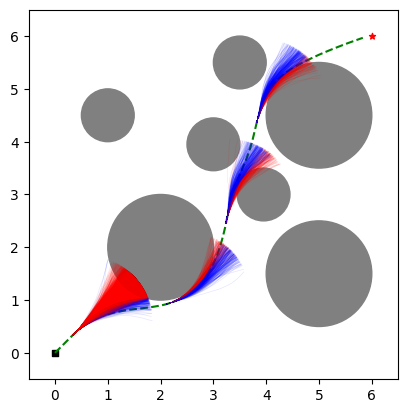}
    \caption{Gaussian, 1 inner iteration} \label{fig:car_cluttered_gaussian_1}
  \end{subfigure}
  \hspace{1mm}
  \begin{subfigure}[t]{0.23\textwidth}
    \includegraphics[scale=0.38]{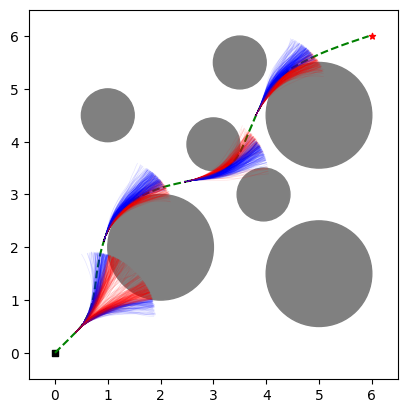}
    \caption{MoG, 1 inner iteration} \label{fig:car_cluttered_mog_1}
  \end{subfigure}\hfill

  \begin{subfigure}[t]{0.23\textwidth}
    \includegraphics[scale=0.38]{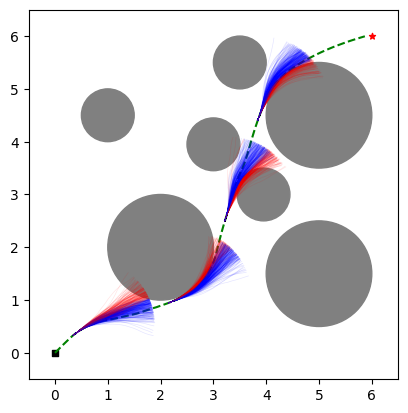}
    \caption{Gaussian, 5 inner iteration} \label{fig:car_cluttered_gaussian_5}
  \end{subfigure}
  \hspace{1mm}
  \begin{subfigure}[t]{0.23\textwidth}
    \includegraphics[scale=0.38]{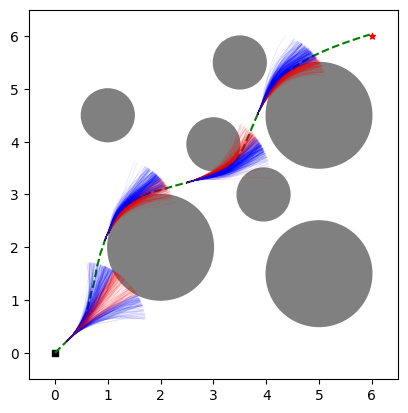}
    \caption{MoG, 5 inner iteration} \label{fig:car_cluttered_mog_5}
  \end{subfigure}\hfill
  \caption{Behavior of Gaussian MPPI and MoG MPPI in a cluttered environment.}
  \label{fig:car-cluttered}
\end{figure}

\begin{table}[t]
\centering
\begin{tabular}{ccccc}
\hline
     & Gaussian,1 & Gaussian,5 & MoG,1 & MoG,5 \\
\hline
Total Samples   & 268288   & 1083392     & \textbf{239616}    & 1092608    \\
Samples Rejected        & 56.11\%   & \textbf{45.65} \%  & 50.85\%   & 46.11\%   \\
\hline
\end{tabular}
\vspace{2mm}
\caption{Total number of sampled trajectories and rejected by MPPI with different distributions and number of iterations}
\label{table:car-cluttered-stats}
\end{table}

\section{Conclusion}


We reinterpret Model Predictive Path Integral (MPPI) control as an
Expectation--Maximization (EM) algorithm, showing that repeated MPPI updates
perform likelihood ascent on a latent-variable optimality model. This view
clarifies the connection between Boltzmann reweighting and weighted
maximum-likelihood projection, enables generalized proposal families beyond
Gaussians, and yields convergence, local-rate, and sufficient-increase results
through classical EM theory. \revise{ Our analysis focuses on the exact expectation-based iteration; finite-sample Monte Carlo effects and receding-horizon closed-loop behavior remain important directions for future work.}

\printbibliography


\newpage
\section{Appendix}
\appendix
\subsection{Proof of \Cref{thm:global-convergence}}
Let
\[
\hat\ell(\theta)
:=
\int p(u;\theta)\exp(-J(u)/\tau)\,du .
\]
We show that \(\hat\ell(\theta)\to 0\) as \(\|\theta\|\to\infty\). Fix
\(M>0\) and decompose
\[
\hat\ell(\theta) \! = \!
\int_{\|u\|\le M}p(u;\theta)e^{-J(u)/\tau}\,du
\!+\!
\int_{\|u\|>M}p(u;\theta)e^{-J(u)/\tau}\,du .
\]
Let
\[
\underline{J}(M):=\inf_{\|u\|\ge M}J(u).
\]
By Assumption~1, \(\underline{J}(M)\to\infty\) as \(M\to\infty\). Therefore,
uniformly in \(\theta\),
\begin{align*}
\int_{\|u\|>M}p(u;\theta)e^{-J(u)/\tau}\,du 
\le&
e^{-\underline{J}(M)/\tau}
\int_{\|u\|>M}p(u;\theta)\,du \\
\le&
e^{-\underline{J}(M)/\tau}.
\end{align*}
Thus the tail term can be made arbitrarily small by choosing \(M\)
sufficiently large.

For the compact-region term, fix \(M\). By Assumption~1.2,
\(p(u;\theta)\to 0\) pointwise for every fixed \(u\) as
\(\|\theta\|\to\infty\). Moreover, by Assumption 1.3,
there exists \(C_M<\infty\) such that
\[
p(u;\theta)e^{-J(u)/\tau}
\le
C_M e^{-J(u)/\tau},
\qquad \|u\|\le M .
\]
Since \(J\) is bounded below, the right-hand side is integrable on
\(\{u:\|u\|\le M\}\). Hence, the dominated convergence theorem gives
\[
\lim_{\|\theta\|\to\infty}
\int_{\|u\|\le M}p(u;\theta)e^{-J(u)/\tau}\,du
=0.
\]
Now let \(\varepsilon>0\). Choose \(M\) such that
\(e^{-\underline{J}(M)/\tau}<\varepsilon/2\), and then choose \(\|\theta\|\) large
enough so that the compact-region integral is less than
\(\varepsilon/2\). Therefore,
\(
\lim_{\|\theta\|\to\infty}\hat\ell(\theta)=0.
\)
Since \(\ell(\theta)=\log\hat\ell(\theta)+\mathrm{const}\), it follows that
\(
\lim_{\|\theta\|\to\infty}\ell(\theta)=-\infty.
\)
Thus the superlevel set
\(
S=\{\theta\in\Omega:\ell(\theta)\ge \ell(\theta_0)\}
\)
is bounded. Since \(\ell\) is continuous and \(\Omega\) is closed, \(S\)
is closed. Hence \(S\) is compact.

\begin{figure*}[t]
    \centering
    \setlength{\tabcolsep}{1pt} 
    \newcommand{\imgcell}[1]{%
        $\vcenter{\hbox{\includegraphics[width=0.25\textwidth]{#1}}}$%
    }
    \begin{tabular}{c c c c}
        \textbf{Iter 10} & \textbf{Iter 12} & \textbf{Iter 14} & \textbf{Iter 16}\\


        \imgcell{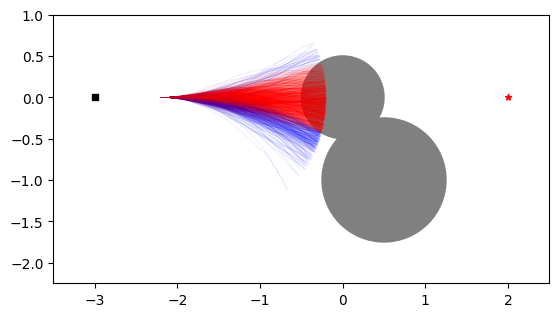} &
        \imgcell{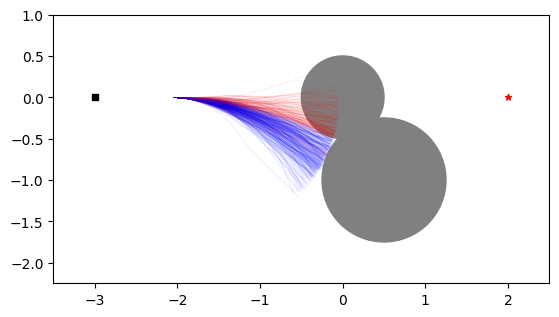} &
        \imgcell{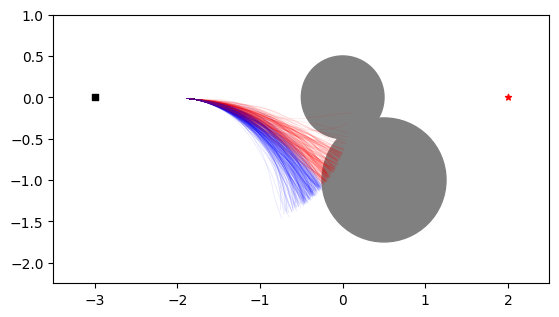} &
        \imgcell{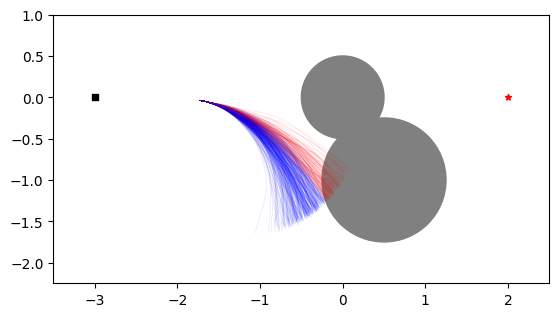} \\

        \imgcell{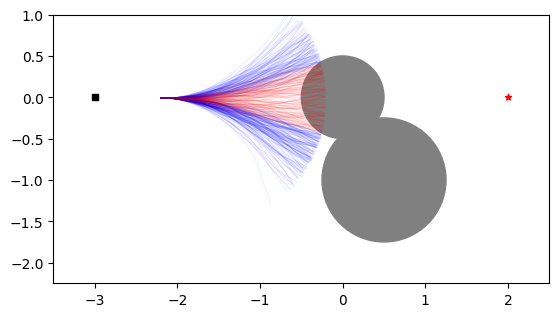} &
        \imgcell{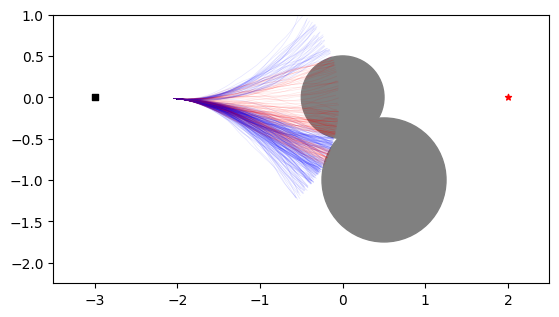} &
        \imgcell{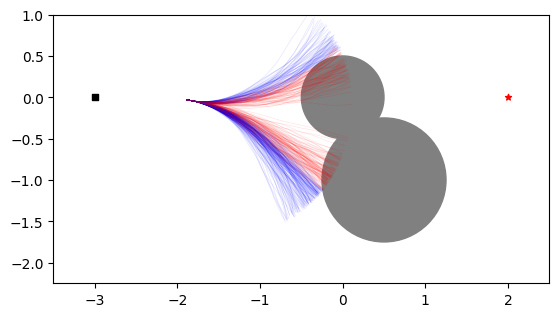} &
        \imgcell{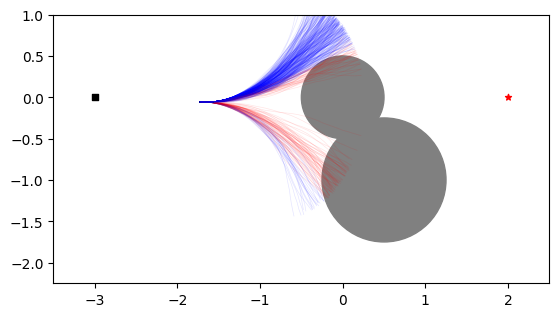} \\
        
    \end{tabular}

    \caption{
        Comparison between Gaussian-MPPI and MoG-MPPI in a cluttered navigation task. The top row shows Gaussian-MPPI, while the bottom row shows MoG-MPPI. At iteration 10, Gaussian-MPPI has already committed to the right-turning mode, whereas MoG-MPPI represents both left- and right-turning modes. Although both methods initially favor the right route, the larger obstacle makes this strategy less favorable at later iterations. MoG-MPPI preserves the left-turning mode, increases its weight, and switches to the left route by iteration 16, while Gaussian MPPI remains concentrated around the initially selected right-turning mode.
    }

    \label{fig:mog}
    \end{figure*}

\subsection{Proof of \Cref{thm:local-linear-convergence}}
\label{app:local_conv}
\begin{proof}
We start with the definition 
\begin{align}\label{eq:EM-fixed-point-appendix}
    \theta_+ &= M(\theta) = \mathrm{argmax}_{\theta_+} Q(\theta_+; \theta) \\
    Q(\theta_+; \theta) &= E_{u \sim q(u;\theta)}[\log p(u;\theta_+)]. \notag
\end{align}
from \eqref{eq:EM-fixed-point}. 

We first derive the Hessian $\nabla^2_{11} Q(\theta_+, \theta)$ and evaluate at $(\theta_+, \theta) = (\theta^*, \theta^*)$.
\begin{align*}
    \nabla^2_{11} Q(\theta_+, \theta) 
    &= \mathbb{E}_{u \sim q(u;\theta)}[\nabla^2_{\theta_+} \log p(u;\theta_+)] \notag \\
    \nabla^2_{11} Q(\theta^*, \theta^*) 
    &= \mathbb{E}_{u \sim q(u;\theta^*)}[\nabla^2_{\theta} \log p(u;\theta^*)].
\end{align*}
We then derive the Hessian $\nabla^2_{12} Q$. 
{\small
\begin{align}\label{eq:hessian-Q12}
    &\nabla^2_{12} Q(\theta_+, \theta) \notag\\
    =& \partial_{\theta} E_{u \sim q(u;\theta)}[
        \nabla_{\theta_+} \log p(u;\theta_+)
    ] \notag\\
    =& \mathbb{E}_{u \sim q(u;\theta)}[
        \nabla_{\theta_+} \log p(u;\theta_+)
        \nabla_{\theta}\log q(u;\theta)^\top
    ] \notag\\
    =& \mathbb{E}_{u \sim q(u;\theta)}[
        \nabla_{\theta_+} \log p(u;\theta_+)
        \nabla_{\theta}\log p(u \mid \mathcal{O}=1;\theta)^\top
    ].
\end{align}
}
In the last step, we use the result $q(u;\theta) = p(u\mid \mathcal{O}=1;\theta)$ form E-step \eqref{eq:E-step-q-gibbs}, and we will use these two interchangeably. To evaluate $\nabla^2_{12}Q(\theta_+, \theta)$ at $(\theta_+, \theta) = (\theta^*, \theta^*)$, we need to use the following intermediate results
\begin{align*}
    &\mathbb{E}_{u \sim p(u\mid \mathcal{O}=1;\theta^*)}[
        \nabla_{\theta} \log p(\mathcal{O}=1,u;\theta^*)
    ] \\
    =& \int p(u\mid \mathcal{O}=1;\theta^*) 
        \frac{
            \nabla_{\theta} p(\mathcal{O}=1,u;\theta^*)
        }{
            p(\mathcal{O}=1,u;\theta^*)
        }
        du \\
    =& \frac{1}{p(\mathcal{O}=1;\theta)} \nabla_{\theta}
        \int p(\mathcal{O}=1,u;\theta^*) du \\
    =& \nabla_{\theta} \log p(\mathcal{O}=1;\theta^*).
\end{align*}
Let $S$ be some random variable with mean $\mu$, we have 
\begin{align*}
    \mathbb{E}[S(S-\mu)^\top] = \mathrm{Cov}[S].
\end{align*}
Let 
\begin{align*}
    S &= \nabla_{\theta} \log p(\mathcal{O}=1, u;\theta^*), \\
    \mu &= \nabla_\theta \log p(\mathcal{O}=1;\theta^*),
\end{align*}
and using the fact that 
\begin{align*}
    &\nabla_\theta \log p(u,\mathcal{O}=1;\theta) \\
    =& \nabla_\theta \log [p(u;\theta) \exp(-\frac{J(u) - J^*}{\tau})] \\
    =&  \nabla_\theta \log p(u;\theta),
\end{align*}
following \eqref{eq:hessian-Q12}, we have 
\begin{align*}
    &\nabla^2_{12} Q(\theta^*, \theta^*) \\
    =& \mathbb{E}_{u \sim q(u;\theta^*)}[
        \nabla_{\theta^*} \log p(u;\theta^*)
        \nabla_{\theta}\log p(u \mid \mathcal{O}=1;\theta^*)^\top
    ] \\
    =& \mathbb{E}_{u \sim q(u;\theta^*)}[S (S - \mu)^\top] \\
    =& \mathrm{Cov}_{u \sim q(u;\theta^*)}[\nabla_\theta \log p(u;\theta^*)].
\end{align*}
Thus
\begin{align*}
    \nabla^2_{11} Q(\theta^*, \theta^*) 
    &= \mathbb{E}_{u \sim q(u;\theta^*)}[\nabla^2_{\theta} \log p(u;\theta^*)], \\
    \nabla^2_{12} Q(\theta^*, \theta^*)
    &= \mathrm{Cov}_{u \sim q(u;\theta^*)}[\nabla_\theta \log p(u;\theta^*)].
\end{align*}
Using the chain rule, it is easy to derive that 
\begin{align*}
    &\nabla^2_{\theta} \log p(\mathcal{O}=1;\theta)  \\
    =& \mathbb{E}_{q(u;\theta)}[\nabla^2_{\theta} \log p(u;\theta)] \\
        &~+ \mathbb{E}_{q(u;\theta)} [\nabla_{\theta} \log p(u;\theta) \nabla_{\theta} \log p(u;\theta)^\top] \\
        &~- \mathbb{E}_{q(u;\theta)} [\nabla_{\theta} \log p(u;\theta)] \mathbb{E}_{q(u;\theta)} [\nabla_{\theta} \log p(u;\theta)]^\top \\
    =& \mathbb{E}_{q(u;\theta)}[\nabla^2_{\theta} \log p(u;\theta)] 
        + \mathrm{Cov}_{u \sim q(u;\theta)}[\nabla_\theta \log p(u;\theta)].
\end{align*}
Evaluating at $\theta = \theta^*$, we have
\begin{align*}
    &\nabla^2_{\theta} \log p(\mathcal{O}=1;\theta^*) \\
    =& \mathbb{E}_{q(u;\theta^*)}[\nabla^2_{\theta} \log p(u;\theta^*)] 
        + \mathrm{Cov}_{u \sim q(u;\theta^*)}[\nabla_\theta \log p(u;\theta^*)] \\
    =& \nabla^2_{11}Q(\theta^*, \theta^*) + \nabla^2_{12}Q(\theta^*, \theta^*).
\end{align*}

For notation simplicity, let 
\begin{align*}
    H = \nabla^2_{11} Q(\theta^*, \theta^*), \\
    C = \nabla^2_{12} Q(\theta^*, \theta^*).
\end{align*}
Since we assume $\theta^*$ is a strict local maximum, we have $\nabla^2_{\theta} \log p(\mathcal{O}=1;\theta) = H + C \prec 0$. Furthermore, since $C$ is a covariance matrix, it must be positive semi-definite, and $H \prec 0$. Thus we have
\begin{align*}
    &\nabla^2_{\theta} \log p(\mathcal{O}=1;\theta)  \\
    =& H + C \\
    =& (-H)^{1/2} \big(-I + (-H)^{-1/2} C  (-H)^{-1/2} \big)(-H)^{1/2}  \\
    \prec& 0,
\end{align*}
from which we can conclude $(-H)^{-1/2} C  (-H)^{-1/2} \prec I$. Since $(-H)^{-1/2} C  (-H)^{-1/2}$ and $-H^{-1}C$ are similar matrices, they have the same spectral radius, thus $\rho(-H^{-1}C) < 1$.

We finally study the fixed iteration \eqref{eq:EM-fixed-point-appendix}. From the optimality condition of \eqref{eq:EM-fixed-point-appendix}, we have
\begin{align*}
    \nabla_1 Q(M(\theta), \theta) = 0 ~\forall \theta,
\end{align*}
by differentiating through the optimality condition with respect to $\theta$, we have 
\begin{align*}
    \nabla_{11} Q(M(\theta), \theta) \partial M(\theta) + \nabla_{12} Q(M(\theta), \theta) = 0 ~\forall \theta.
\end{align*}
At the fixed point $\theta = \theta^*$, we have $M(\theta^*) = \theta^*$, thus
\begin{align*}
    \partial M(\theta^*) = - \nabla^2_{11} Q(\theta^*, \theta^*)^{-1} \nabla^2_{12} Q(\theta^*, \theta^*) = -H^{-1}C,
\end{align*}
and $\rho(\partial M(\theta^*)) < 1$ from previous results. Then the convergence result
follows directly from \cite[Sec~3.9]{mclachlan2008algorithm} and Ostrowski Theorem \cite[Thm 10.1.3]{ortega2000iterative}

\end{proof}

\subsection{Proof of \Cref{thm:exp-family-sufficient-increase}}
\label{app:sufficient_increase}
\begin{proof}
Define
\begin{align*}
    g(\eta_1, \eta_2) \!
    = \! \mathbb{E}_{u \sim q(u; \eta_2)}
        \bigg[ 
            \log 
            \frac{
                h(u) \exp[\eta_1^\top T(u) \! - \! A(\eta_1) \! - \! \frac{J(u)}{\tau}]
            }{
                q(u; \eta_2)
            }
        \bigg],
\end{align*}
which is the ELBO of $\ell(\eta_1)$ evaluated at $q(\cdot;\eta_2)$. Hence
\begin{align*}
    \ell(\eta_1) = g(\eta_1, \eta_1) \geq g(\eta_1, \eta)
    \quad \forall \eta.
\end{align*}
Therefore,
\begin{align*}
    \ell(\eta_+) - \ell(\eta)
    &= g(\eta_+, \eta_+) - g(\eta, \eta) \\
    &\ge g(\eta_+, \eta) - g(\eta, \eta) \\
    &= \mathbb{E}_{u \sim q(u; \eta)}
        \big[ 
            (\eta_+ - \eta)^\top T(u) - A(\eta_+) + A(\eta)
        \big] \\
    &= (\eta_+ - \eta)^\top \mathbb{E}_{q(u;\eta)}[T(u)] - A(\eta_+) + A(\eta).
\end{align*}
Using the moment-matching condition \eqref{eq:exp-family-m-step-moment-matching},
\begin{align*}
    \ell(\eta_+) - \ell(\eta)
    &\ge A(\eta) - A(\eta_+) - (\eta - \eta_+)^\top \nabla_\eta A(\eta_+) \\
    &\ge \frac{\alpha}{2}\|\eta_+ - \eta\|_2^2,
\end{align*}
where the last inequality follows from the $\alpha$-strong convexity of $A(\eta)$.
\end{proof}

\subsection{Additional Experiments}

To further isolate the benefit of multi-modality beyond simple symmetry-breaking, we conducted the additional experiment shown in \Cref{fig:mog}. At iteration 10, Gaussian MPPI has already collapsed to the right-turning mode, whereas MoG-MPPI maintains two distinct trajectory modes corresponding to the left- and right-turning strategies. By iteration 12, both methods assign larger weights to the right-turning direction; however, MoG-MPPI still preserves a non-negligible left-turning component. 
At iteration 14, the enlarged obstacle causes the right-turning strategy to incur a higher cost than the left-turning alternative. As a result, MoG-MPPI reallocates probability mass toward the left-turning component and eventually commits to the left route at iteration 16. In contrast, once Gaussian MPPI has collapsed to the right-turning mode, the unimodal distribution struggles to recover and transition to the alternative strategy.

This experiment demonstrates that the advantage of multi-modality is not limited to obstacle avoidance or symmetry-breaking. Rather, maintaining multiple modes enables the controller to preserve and adaptively reweight several feasible strategies before making a final commitment.


\end{document}